\documentclass[pmlr]{jmlr}% new name PMLR (Proceedings of Machine Learning Research)

\usepackage{longtable}% for long tables

\usepackage{booktabs}
\usepackage[load-configurations=version-1]{siunitx} % newer version
\newtheorem{customtheorem}{Theorem}
\newtheorem{customdefinition}{Definition}
\newtheorem{customcorollary}{Corollary}

\newtheorem{customlemma}{Lemma}

\title[PAC Learning PDFA]{PAC learning PDFA from data streams}

\author{\Name{Robert Baumgartner} \Email{r.baumgartner-1@tudelft.nl}\\ 
   \Name{Sicco Verwer} \Email{s.e.verwer@tudelft.nl}\\
   \addr Department of Software Technology \\
         Delft University of Technology \\
         Delft, The Netherlands \\
   }
\begin{document}

\maketitle

\begin{abstract}
This is an extended version of our publication \textit{Learning state machines from data streams: A generic strategy and an improved heuristic, International Conference on Grammatical Inference (ICGI) 2023, Rabat, Morocco} \citep{icgi_version}. It has been extended with a run-time and memory analysis, a formal proof on PAC-bounds\footnote{An accessible overview of the main proof ideas is available at \\ \url{https://robba.github.io/pac-bounds}} in Section \ref{sec:analysis}, and the discussion and analysis of a similar approach has been moved from the appendix and now has a full dedicated section (\ref{sec:space_save_discussion}).

\vspace{\baselineskip}

\noindent State machine models are models that simulate the behavior of discrete event systems, capable of representing systems such as software systems, network interactions, and control systems, and have been researched extensively. The nature of most learning algorithms however is the assumption that all data be available at the beginning of the algorithm, and little research has been done in learning state machines from streaming data. In this paper, we want to close this gap further by presenting a generic method for learning state machines from data streams, as well as a merge heuristic that uses sketches to account for incomplete prefix trees. We implement our approach in an open-source state merging library\footnote{Available at \url{https://github.com/tudelft-cda-lab/FlexFringe}} and compare it with existing methods. We show the effectiveness of our approach with respect to run-time, memory consumption, and quality of results on a well known open dataset. Additionally, we provide a formal analysis of our algorithm, showing that it is capable of learning within the PAC framework, and show a theoretical improvement to increase run-time, without sacrificing correctness of the algorithm in larger sample sizes.
\end{abstract}
\begin{keywords}
State machine learning, Automata learning, Streaming data, PAC-Learning
\end{keywords}

\thispagestyle{empty} % Removes header/footer from the first page

\section{Introduction}
\label{sec:intro}

State machines are insightful models that naturally represent formal languages and discrete systems, and have been extensively used in various domains such as model checking \citep{model_checking_book} and modeling discrete event systems such as truck driving \citep{sicco_thesis}, computer networks \citep{pellegrino_2017}, and controllers and software systems \citep{walkinshaw_2016}. A major advantage of state machines is that, combined with expert knowledge, they are interpretable \citep{chris_interpreting}. Learning state machines can be roughly subdivided into active learning and passive learning. Active learning learns from interactions with a system under test. Passive learning learns directly from observed execution traces without interfering with a system's execution. Most passive algorithms require all data to be available before running the inference step and hence cannot learn from data streams. Exceptions are the works of \cite{balle_2012, balle2014} and \cite{schmidt2014online}. These works employ a typical streaming approach where decisions made by the learning algorithm are postponed until sufficient data has been observed. In this work, we propose a novel streaming strategy, with its main advantage being that it corrects errors made during previous iterations. In this way, we can perform learning steps even when the statistical tests performed by the algorithm are inconclusive: subsequent iterations can correct these steps. As a consequence, we learn faster. Our streaming strategy is generic and can be used with multiple heuristics.

Like most passive learning algorithms, we adopt state-merging in the red-blue framework from \cite{lang_1998}. This framework starts with a prefix-tree that directly encodes the entire input data set and then greedily merges states of this tree until no more consistent merges can be performed. Being a streaming algorithm, our method only keeps counts from all observed traces in memory and builds the prefix tree only for states with sufficient occurrence counts. To learn models from such a partially specified prefix tree, we propose a new Count-Min-Sketch based merge heuristic. The main idea is that we hash future prefixes and store them in sketches in each state. Our new heuristic performs consistency checks directly on these sketches. Importantly, Count-Min-Sketches allow for efficient updates which allows performing and undoing of merges in constant time.

We implemented our approach in the open source FlexFringe library \citep{flexfringe}, a flexible state-merging framework that allows to specify one's own heuristics and consistency checks. We added streaming capability and our own consistency check using Count-Min-Sketches. FlexFringe provides efficient structures for performing and undoing merges, as well as a multitude of heuristics such as Alergia~\citep{alergia_1994} that we use in our experiments. We demonstrate the effectiveness of our approach on the well known PAutomaC dataset. We evaluate on run-time, memory footprint and approximation quality. We experiment in settings where merge undoing is enabled, and the more common streaming setting where merges are postponed until sufficient data is available. In our experiments our method clearly outperforms the earlier approach of \cite{balle2014} in terms of approximation quality. It also compares favorably to streaming using Alergia as the heuristic without merge undoing and performs better on most problems when undoing merges is enabled, albeit the gap has shrank compared to the not-undoing setting. The results clearly demonstrate that the ability to correct mistakes is crucial for obtaining good performance using traditional merge heuristics such as Alergia. The Count-Min-Sketches result in improved consistency checks when postponing merges, but even these benefit from merge undoing. Our method gets close to but does not reach the performance of a non-streaming version of Alergia that simply loads all the available data in the prefix tree.

\section{Related Work}\label{section:related_work}

There are two main approaches to learning state machines, either by active learning \citep{lstar, queries, vaandrager2017model}, or via state merging \citep{learning_grammars}. While active learning requires the availability of an oracle being able to answer queries, state merging learns from example traces. Example traces are being represented via a complete tree, and then greedily minimized in accordance to Occam's razor. These all perform a consistency test of where a pair of state is mergeable and compute a heuristic value to determine which merge to perform first. For example, Alergia \citep{alergia_1994} learns state machines representing probabilistic languages via sampled distributions, and k-Tails~\citep{ktails} compares the subtrees of state pairs. Other seminal works include the  RPNI algorithm \citep{rpni_paper} and the EDSM algorithm \citep{lang_1998}. Several types of models can be learned like this. \cite{verwer_rti} learn timed automata from timed strings, \cite{walkinshaw_2016} learn extended finite state machines to represent software and control systems, and \cite{gkplus} learn guarded finite state machines. Hybrid state machine models can be learned by taking different aspects of a system into account \citep{vodencarevic_2011, lin2018moha}. Since the problem of finding an exact solution in passive learning has been shown to be NP-hard \citep{complexity}, several search strategies have been developed \citep{search1, search2, search3}. Different ways to speed up the main algorithm have also been proposed through divide and conquer and parallel processing \citep{luo2017inferring, akram2010psma, shin2021prins}.

Few works deal with making state-merging algorithms more scalable when run on streaming data. \cite{schmidt2014online} utilize frequent pattern data stream techniques to tackle the problem. \cite{rpni2paper} build on the RPNI algorithm and learn state machines from positive and negative examples in an incremental fashion. Conflicts that can arise from new unseen negative examples are resolved via splitting states until conflicts are resolved. \cite{balle_2012} present a theoretical work of streaming state machines using modified space saving sketches \citep{metwally}, and extend it with a parameter search strategy in \citep{balle2014}. \cite{schouten2018} implements a streamed merging method in the Apache framework, also using the Count-Min-Sketch data structure \citep{count-min-sketch}. In this paper, we present a new streaming learning method that, like the algorithm of \cite{rpni2paper}, can correct mistakes from earlier iterations, and like the algorithm of \cite{balle_2012} uses sketches to approximate the information contained in the observed traces. 

\section{Background}\label{section:background}

A PDFA is a tuple defined by $\mathcal{A} = \{\Sigma, Q, q_0, \tau, \lambda, \eta, \xi\}$, where $\Sigma$ is a finite alphabet, $Q$ is a finite set of states, $q_0 \in Q$ is a unique starting state, $\tau : Q \times \Sigma \cup \{ \zeta \} \to Q$ is the transition function with $\zeta$ the empty string, $\lambda : Q \times \Sigma \to \left[0,1\right]$ is the symbol probability function, and $\eta : Q \to \left[0,1\right]$ is the final probability function, such that $\eta(q) + \sum_{a \in \Sigma} \lambda(q,a) = 1$ for all $q \in Q$. $\xi \not \in \Sigma$ denotes the final symbol, indicating the end of a sequence. 

In the following we will stick to the following convention: We denote the Kleene star operation over $\Sigma$ by $\Sigma^*$, and denote $x\Sigma^*$ the set of all possible strings with prefix $x$. When we decompose a string into multiple parts or when referring to general strings over $\Sigma^*$ we use the letter $\sigma$. Single elements of the alphabet are denoted by the character $a$, s.t. $a \in \Sigma$. For example, the string $\sigma_1 a \sigma_2$ is the string obtained through the concatenation of string $\sigma_1$, the symbol $a$ and the string $\sigma_2$. Given string $\sigma=a^1 a^2 \ldots a^n$, we call the sequence $a^{i+1} a^{i+2} \ldots a^{i+m}$, $i \in [0, n-m]$ a substring of length $m$ of $\sigma$. Given transition function $\tau$, a string traverses the PDFA recursively in order $\tau(q_i, a^{i+1}) = q_{i+1}$ for all $0 \leq i \leq n-1$. For convenience we define a traversal through an entire string or substring $\sigma$ via shorthand notation $\tau(\sigma)$. We consider a parent of a state $q$ a state $q'$ s.t. $\exists a \in \Sigma: \tau(q', a)=q$. In the prefix tree each node has exactly one parent.

The probability $P( \sigma)$ of a string $\sigma = a^1, \ldots, a^n$ can be computed via $P( \sigma) = \lambda(q_0,a^1) \cdot \lambda(q_1, a^2) \cdot \ldots \cdot \lambda(q_{n-1}, a^n) \cdot \eta(q_n)$. In this work we model sequential information via sampled distributions over strings emanating from a state $q_i$ via $D_{q_i}(\sigma), \sigma=a^{i+1} a^{i+2}...$, where $D_q(\sigma)\rightarrow \left[0, 1\right]$ models a sampled distribution. We call a (sub-)string $\sigma_2$ an outgoing string from state $q$ if $\exists \sigma' = \sigma_1\sigma_2 : \tau(\sigma_1)=q$. We further define the size $s_q$ of a given node $q$ by the number of input strings $\sigma \in I$ from input $I$ during the learning that traverse $q$, i.e. $n_q = |\{ \sigma \in I | \sigma=a_1a_2...a_{|\sigma|} \wedge \exists i \in [1, |\sigma|]: \tau(q_{i-1}, a^i)=q \}|$. The size of the root node $q_0$ is $n_{q_0}=|\{ \sigma \in I \}|$.

PDFAs can be learned via state-merging. In a first step, a state-merging algorithm constructs a prefix tree representing the input data completely, meaning that every state in this tree has only one unique access sequence. Algorithms differ in how they minimize this tree. Usually heuristics are employed to greedily merge equivalent state pairs $(q, q')$. In each iteration the heuristic will assign a score $\phi$ to every mergeable state pair (\textit{candidate pair}), and the merge with the highest score will be performed. After merging two states into one, a subroutine is started such that $\tau$ is deterministic, i.e. that there is only one possible transition $\tau(q', a')$ for all $q'\in Q \text{, } a' \in \Sigma$. Multiple search strategies exist to limit the search of state pairs. In this work we stick close to the red-blue-framework \citep{lang_1998}. The first red state is the root node of the prefix tree, and blue states are all states emanating from red states. Only pairs of one red and one blue state can be merged, and the state resulting from a merge is always red. In case no merge can be performed, the blue state $q$ with the largest size $s_q$ will be turned red. Note that in this framework nodes that are not red always have exactly one parent, and the parent of a blue node is always red.

\section{Methodology}\label{section:methodology}

\subsection{Merge heuristic}\label{subsec:heuristic}

The core idea of our merge heuristic is to store the counts of outgoing strings of each state. Because this quantity can become very large in data streams we use the Count-Min-Sketch (CMS) \citep{count-min-sketch} data structure, and we equip each state with one such CMS. We store each state's outgoing strings in its sketch. Inserting elements and retrieving them can be done in time $\mathcal{O}(1)$, and only depend on the size of the sketches. Two CMS can be considered as matrices, and states can easily be merged and unmerged via matrix addition and subtraction. To compare two sketches we retrieve the counts of all seen strings $\sigma$ and model the distribution over those as frequencies. We can then compare these two distributions via statistical tests. A problem that arises is that the number of possible strings grows $\mathcal{O}(|\Sigma|^{F_s})$, where we denote as $F_s$ the maximum length of the strings that we consider (a hyperparameter). We propose two solutions to tackle this problem: The first one by constructing multiple sketches per state, one for each possible size of string. In this setting, the first sketch would only store strings of size $1$, the second sketch strings of size $2$, and so on. The second solution is concerning the hash-function that is utilized in the CMS. We describe the heuristic in more detail in the following subsection.

\subsubsection{Count-Min-Sketches and the heuristic}

We consider two states behaving similar if their multiset of outgoing strings is similar. We consider the regular set of outgoing strings of a state $q$ the set $\{ \sigma \in \Sigma^* | \exists x \in \Sigma^* \wedge \sigma'=\sigma: \tau(x)=q \}$. The multiset of outgoing strings of state $q$ is simply the respective multiset, that also stores the number of times each of the elements has been attempted to be inserted into the set. We denote the count of an arbitrary element $y$ via $c_y$.

Given state $q$, we assume a sampled distribution $D_q: \sigma \rightarrow [0, 1]$, again with $\sigma \in \Sigma^*$\footnote{Recount that the Kleene-star operation applied to $\Sigma$, i.e. $\Sigma^*$, denotes the set of finite strings. Hence the set of all $\{ \sigma \in \Sigma^* \}$ is a finite set and we can obtain discrete probabilities for the occurrence of each string.}. For us, $D_q(\sigma_i)$ is simply the frequency of string $\sigma_i$, that is $D_q(\sigma_i)=\frac{c_{\sigma_i}}{n_q}$. Because the set $\Sigma^*$ can be potentially very large, we approximate $D_q$ for each state via a variant of the Count-Min-Sketch (CMS) data structure \citep{count-min-sketch}, which is why we call our heuristic CSS (\underline{C}MS-based \underline{S}pace \underline{S}aving). 

Formally, a CMS is a probabilistic data structure to summarize data streams. In practice a CMS is a matrix represented by $d=\left\lceil{\ln \frac{1}{\gamma}}\right\rceil$ rows and $w=\left\lceil{\frac{e}{\beta}}\right\rceil$ columns, where $e$ is the basis to the natural logarithm $\ln$. Given counts $c_{y_i}$ of elements $y_i$ of a given set of possible events $\mathcal{Y}$, the CMS is able to store the counts and retrieve them using $\mathcal{O}\left( \frac{1}{\beta} \right)$ space and $\mathcal{O} \left( \frac{1}{\gamma} \right)$ time. To do so each row of the CMS is associated with a hash-function $h$ mapping elements $y$ of set $\mathcal{Y}$, $y \in \mathcal{Y}$,  to $h: y\rightarrow \mathbb{N} \cap [0, w-1]$. That means in total there exist $d$ hash functions, and we want them to be i.i.d. chosen from a pairwise-independent hash-family $\mathcal{H}$ \citep{count-min-sketch}. We denote $h_j$ the hash function associated with row $j$, and for convenience we define an inverse mapping $h_j^{-1}(y)$. At initialization, each entry of the CMS is set to zero. To store an incoming element $y$ the CMS hashes $y$ for each row $j$ and increments the row at $h_j(y)$ by 1. Retrieving an approximated count of a given element $y$ works via again hashing the $y$ once per row. The approximated count is the minimum $\min_{j \in [1, d]} h_j^{-1}(y)$. The CMS is then able to retrieve an approximated count $\hat{c}_{y_i}$ of $y_i$ via the $retrieve(y_i)$ operation. The error bound is $\hat{c}_{y_i} \leq c_{y_i} + \beta \sum_i c_{y_i}$, which holds with probability at least $1-\gamma$. 

As already mentioned we count strings $\sigma \in \Sigma^*$, and we denote the count of string $\sigma_i$ by $c_{\sigma_i}$ and its approximated count by $\hat{c}_{\sigma_i}$. In addition to conventional CMS our data structure one extra attribute and two extra operations. The extra attribute is a counter for the final counts $\xi$, indicating that a sequence did end in the previous state. Our sketches further support a $+$ and a $-$ operation, which we define as the sum and subtraction of two matrices of equal size, and the same for the final counts. These two operations will be used to perform and undo merges. In order test whether two states behave similar we perform the statistical test from the works of \cite{alergia_1994}, which they use in their own Alergia-algorithm. This statistical test, herein after referred to as Alergia-test, is used to check whether two sampled distributions are similar. State-pairs that  pass this test get assigned a score $\phi$. To this end we use the Cosine-similarity. Given two vectors  $v_1$ and $v_2$ the Cosine-similarity measures the angle in between the two vectors via 

\begin{equation}
    cosine\_similarity(v_1, v_2) = \frac{v_1 \cdot v_2}{||v_1||_2 ||v_2||_2},
\end{equation}

where $v_1 \cdot v_2$ is the dot-product of the two vectors, and $||v_i||_2$ denotes the L2-norm of vector $v_i$.  We show the entire subroutine including the Alergia-test in \algorithmref{alg:consistency}. In this subroutine $\alpha$ is a hyperparameter to be set before the start of the learning algorithm. It represents a bound on the probability of a wrong rejection, which is upper bounded with $2\alpha$ (a wrong rejection means that the two distributions are similar, but the test deems them dissimilar).

\begin{algorithm2e}[t]
  \label{alg:consistency}
  \caption{$Consistency-routine$}
  \KwIn{\\
  $CMS_{q_1}$, $CMS_{q_2}$: Count-Min-Sketches of state $q_1$ and $q_2$ respectively \\ 
  $n_{q_1}$, $n_{q_2}$: Size of state $q_1$ and $q_2$ respectively \\
  $S$: The set of all strings $x$ observed so far.
  $\alpha$: Hyperparameter to be set.
  }
  \KwOut{Boolean value indicating consistency, score $\phi$ if applicable}

  $v_1$, $v_2 \leftarrow$ empty lists\\
  \ForEach{$\sigma \in S$}{
    $\hat{c}_{\sigma_{q_1}} \leftarrow CMS_{q_1}.retrieve(\sigma)$\\
    $\hat{c}_{\sigma_{q_2}} \leftarrow CMS_{q_2}.retrieve(\sigma)$\\
    \uIf{$\left| \frac{\hat{c}_{\sigma_{q_1}}}{n_{q_1}} - \frac{\hat{c}_{\sigma_{q_2}}}{n_{q_2}} \right| > \sqrt{\frac{1}{2}log\left(\frac{2}{\alpha}\right)}\left(\frac{1}{\sqrt{n_{q_1}}}+\frac{1}{\sqrt{n_{q_2}}}\right)$}{
      \Return{false}\\
    }
    $v_1.append(\hat{c}_{\sigma_{q_1}} / n_{q_1})$, $v_2.append(\hat{c}_{\sigma_{q_2}} / n_{q_2})$\\
  }
  \Return{$true$, $cosine\_similarity(v_1, v_2)$} \\
\end{algorithm2e}

\subsubsection{The runtime problem}\label{sec:hashing}

In practice we have to bound the size of outgoing strings of states, because they can be arbitrarily long. We denote the maximum length of an outgoing (sub-)string for any state $q \in Q$ by $F_s$. The run-time problem then is that the size of set $S$ from \algorithmref{alg:consistency} does grow with $\mathcal{O}(|\Sigma|^{F_s})$. We propose two solutions to overcome this problem. The first one is a simple decoupling. Instead of storing all strings in one sketch, we construct $F_s$ sketches per state, namely $C_1, C_2, \ldots C_{F_s}$, and each of them stores one size of string in ascending order: (sub-)strings of length $1$ $\rightarrow C_1$, (sub-)strings of length $2$ $\rightarrow C_2$, ..., (sub-)strings of length $F_s$ $\rightarrow C_{F_s}$. Consistency checks are then performed in ascending order on those CMS, and higher order CMS are only checked if all lower order CMS passed the test already. The score $\phi$ is simply the average of all scores $\phi_1 \ldots \phi_{F_s}$ returned by those checks. \figureref{fig:sketch_future} illustrates an example, where state $S13$ stores sketches with $F_s=3$.

While this helps us reduce runtime in practice at the cost of some memory, the worst case run-time remains at $\mathcal{O}(|\Sigma|^{F_s})$. This can practically be a problem on data where a large $F_s$ is required to distinguish states properly. In this case we propose a second solution.  We adopt the ideas introduced by Locality Sensitive Hashing (LSH) \citep{original_lsh, p_stable_lsh}. Intuitively, LSH hashes elements of some domain $W$ of dimensionality $r_1$ into another domain $U$ of dimensionality $r_2 < r_1$ while preserving a distance measure of choice. The hashing-algorithm is based on the chosen distance measure. Because we are dealing with set of discrete symbols $a\in \Sigma$ we choose to preserve the Jaccard similarity of two strings $\sigma_1$ and $\sigma_2$, denoted by $Jaccard(\sigma_1, \sigma_2)$, and hash the strings using the MinHash-algorithm \citep{min_hash_paper, lsh_lecture}. Taking a hash-function $h_i$ from hash-family $\mathcal{H} = \{h: \mathcal{W}\rightarrow \mathcal{U}\}$ we obtain $P(h_i(\sigma_1) = h_i(\sigma_2)) = Jaccard(\sigma_1, \sigma_2)$, where $Jaccard(\sigma_1, \sigma_2) = \frac{|{\{ \sigma_1 \} \cap \{ \sigma_2 \}}|}{|\Sigma|}$. In this notation we use the shorthand-notation $\{ \sigma_i \}$ to denote the set $\{ a'\in \Sigma | \sigma_i = a^1a^2 \ldots a^n \text{, } \exists j \in [1, n]: a^j=a' \}$. In other words, the more similar two strings, the more likely they are being hashed into the same bin. We introduce a new hyperparameter $l_{m}$, and define $l_{m}$ mutually-independent hash-functions $h_i, i\in [1, l_{m}]$. Because our goal is to reduce run-time we only need to hash strings with length larger than $l_{m}$, hence the run-time is then $\mathcal{O}(|\Sigma| + |\Sigma|^{l_m})$, where the $|\Sigma|$ term stems from the maximum number of attempts it takes to hash a string of length $F_S$. \figureref{fig:css_minhash_explained} shows the approach with a hashing down to size $2$. We herein after call this approach \textit{CSS-MinHash}.

\begin{figure}[h]%[H]%[htbp]
    \centering
    \subfigure[An example of how the futures are stored within the sketches. The numbers are purely fictional.]{\label{fig:sketch_future}
      \includegraphics[width=0.45\textwidth]{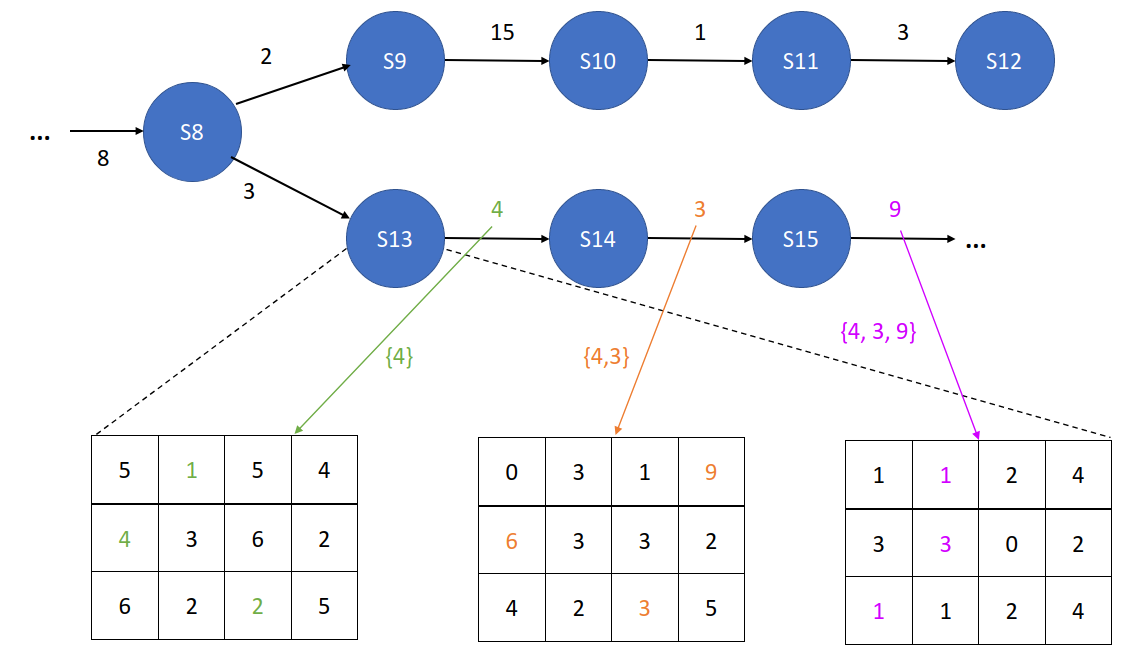}}
    \qquad
    \subfigure[Illustration of CSS-MinHash. Here we hash down to a size of $2$.]{\label{fig:css_minhash_explained}
      \includegraphics[width=0.45\textwidth]{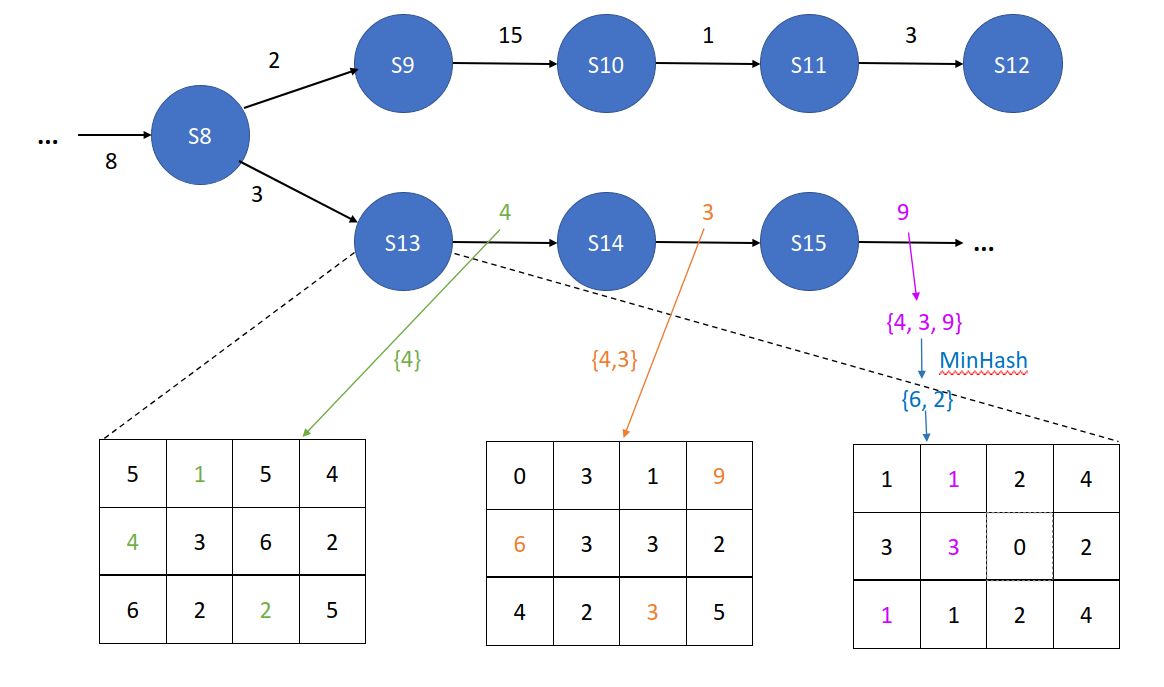}}
    \caption{The two solutions to the problem of uniform distributions in the sketches.}
    %\label{fig:csm_count}
\end{figure}

\subsection{Streaming}\label{subsec:streaming}

Our streaming algorithm consists of two main ideas. The first one is that we discard information that does not occur often enough. This is common practice in streaming algorithms. To do so we use the red-blue framework and include a threshold $t_S$, where states can only be created with red or blue states as parent nodes, and a state's size $n_q$ has to exceed $t_S$ to be able to become a blue node. The second idea is to undo and redo merges. In our approach we divide the incoming data stream into batches of size $B$. After reading $B$ sequences from the stream, we first perform a greedy minimization step. We keep track of all the operations (herein after called \textit{refinements}) we performed on the prefix tree. A refinement is either the identification of a new red state or the merge of a state-pair. Once we cannot perform refinements anymore, we save the resulting automaton and reverse all the refinements we did, starting anew. In the next iteration we will first try to perform all the refinements from the previous iteration again, and then perform greedy minimization. 

Because we can logically separate the two main steps we perform, namely the batch-wise streaming of the prefix tree and the subsequent minimization routine, we have split those up into two algorithms. We describe the two steps individually in the next two subsections.

\subsubsection{Streaming the prefix tree}

Streaming the prefix tree starts with the root node marked as a red state. We read in $B$ elements of the data stream, and create states emanating from the root node. Those states are neither red nor blue, unless they have been accessed by an incoming sequence at least $t_S$ times, in which case we mark them blue. States can only be created emanating from red or blue states, thus in each iteration the fringe of the prefix tree can grow at a maximum of two more states from each existing state. We do this to only save relevant information of the common behavior of the system. Every time a node $q$ gets traversed by a string $\sigma \in I$ we update it with data. This operation depends on the merge heuristic that we chose at the start of the program. In case of our CSS-heuristic we update the node with its outgoing strings.

Once a batch is completed a minimization routine, described in section \ref{sec:minimization_routine}, will start. The minimization routine returns a hypothesis automaton and marks all the states that have been red or blue as a result of the minimization routine. We then start reading in another batch as described before, and the minimization routine starts again. The algorithm is depicted in \algorithmref{alg:streamingtree}.

\begin{algorithm2e}[t]
  \label{alg:streamingtree}
  \caption{Streaming the tree}
  \KwIn{Stream of sequences $I$, batch-size $B$, threshold $t_S$, upper bound on states $n$}
  \KwOut{Hypothesis automaton $H$}
  
  $H \leftarrow$ root node\\
  $c \leftarrow 0$\\
  $R_{old}\leftarrow$ empty queue\\
  \ForEach{$\sigma_i \in X$}{
    $\sigma_i \leftarrow \sigma_i\xi$ \\
    $q \leftarrow$ root node\\\
    $n_q \leftarrow n_q + 1$\\
    update $q$ with relevant data\\
    \ForEach{$j \in len(\sigma_i)$}{
        $a^j \leftarrow \sigma_i[j]$\\
        \uIf{transition from $q$ with $a^j$ exists}{
            $q \leftarrow \tau(q, a^j)$\\
        }
        \uElseIf{$q$ is red or $q$ is blue}{
            create new node $q'$\\ 
            set $\tau(q, a_j)=q'$\\
            $q \leftarrow q'$\\
        }
        \uElseIf{$q.parent$ is red and $n_q \geq t_S$}{
            mark $q$ blue\\
        }
        \uElse{
            break inner loop \\
        }
        %\uEndIf
        $n_q \leftarrow n_q + 1$\\
        update $q$ with relevant data\\
    }
    $c \leftarrow c + 1$\\
    \uIf{$c == B$}{
        $R_{old} \leftarrow perform\_minimization\_routine(H, R_{old})$\\
        $c\leftarrow 0$\\
    }
    \uIf{size of last $H = n$ \text{or} I empty}{
        $perform\_minimization\_routine(H, R_{old})$\\
        \Return $H$\\
    }
  }
\end{algorithm2e}

\subsubsection{Minimization routine}\label{sec:minimization_routine}

The core idea of our minimization routine is that, depending on the implementation, undoing and performing refinements again can be done cheaply. Performing merges and undoing them can be done cheaply in constant run-time in Flexfringe \citep{flexfringe}. The advantage of undoing and redoing refinements is that the learner gets the chance to correct mistakes earlier. This way we can return a model earlier without compromising its future correctness. The first time we start the minimization routine it performs the normal minimization routine described in section \ref{section:background}. What is new is that we store all the refinements that are performed in this step in a queue $R_{new}$. We save the found automaton as hypothesis $H$ and undo all refinements that we did. At the end of the subroutine we  we save $R_{new}$ as another queue $R_{old}$ and return to the prefix-tree streaming subroutine.

From the second time onward, when we enter the minimization subroutine we will start with $R_{old}$ and retry every refinement we performed in the previous run. Assuming that the underlying distribution of strings from the data stream did not change most of these refinements will still hold, hence this approach saves us run-time. The way we deal with the case that a refinement is not possible anymore is based on the following observation. We identified two main causes why a refinement cannot take place anymore:
\begin{enumerate}
    \item Consistency-failure: The underlying structure of the node has changed according to the merge heuristic. It could be for instance that the distribution of strings that passed it changed significantly.
    \item Structural-failure: The underlying structure of the tree has changed. This can be caused by previous consistency failures, in which case the sub-tree might change.
\end{enumerate}

While consistency-failures cannot be fixed in this batch, we follow the intuition that refinements with structural failures can often still hold after obtaining an appropriate tree structure. As an example, say that an earlier performed merge refinement of two states $q_1$ and $q_2$ cannot be performed, because state $q_2$ is neither red nor blue anymore. We then continue the minimization procedure and later reconsider what to do with $q_2$. If at that time $q_2$ is blue, we perform the merge if it is valid. The advantage of this strategy is that it avoids having to recompute consistency checks where unnecessary. Compared with structural checks, which can be done in a simple flag check, consistency check are much more expensive. Our strategy works then as follows. 

We first test on structural failures on each refinement $r\in R_{old}$. In case a structural failure did not happen, we test on a consistency failure. We perform the refinement if none of those two failures did occur. If a consistency failure occurred we discard $r$ and test for the next refinement from $R_{old}$. If however only a structural failure occurred, we push $r$ onto another queue $R_{failed}$. Once we've exhausted $R_{old}$ we perform the greedy minimization again. This time however, each time a new refinement is identified we iterate once over $R_{failed}$ and test again for both failure modes. If no failure can be detected we perform the refinement. Just as in the first iteration we store all the refinements performed in $R_{new}$ and save it as $R_{old}$ at the end of the subroutine. The whole procedure is described in \algorithmref{alg:mergestrategy}.
\begin{algorithm2e}[t]
  \label{alg:mergestrategy}
  \caption{Minimization routine}
  \KwIn{Current hypothesis $H$, queue $R_{old}$}
  \KwOut{queue $R_{new}$}

  $R \leftarrow$ empty stack\\
  $R_{failed}$, $R_{new} \leftarrow$ empty queue\\
  \While{$R_{old}$ not empty}{
    $r\leftarrow R_{old}.pop()$\\
    \uIf{$r$ possible via structure and via consistency}{
        perform $r$ on $H$\\
        $R.push(r)$\\
        $R_{new}.push(r)$\\    
    }
    \uElseIf{$r$ possible via structure}{
        $R_{failed}.push(r)$\\
    }
  }
  \While{H contains at least one blue or white state}{
    Select best possible refinement $r$\\
    perform $r$  on $H$\\
    $R.push(r)$\\
    $R_{new}.push(r)$\\
    
    \ForEach{$r \in R_{failed}$}{
        \uIf{$r$ possible via structure and possible via consistency}{
            $R_{failed}.delete(r)$\\ \label{line:queuedeletion}
            perform $r$ on $H$\\
            $R.push(r)$\\
            $R_{new}.push(r)$\\
        }
    }
  }

  Save current $H$\\ \label{line:saveH}
  \While{$R$ not empty}{
    $r \leftarrow R.pop()$\\
    undo $r$ on $H$\\
  }

  \Return $R_{new}$\\
\end{algorithm2e}

\section{Experiments on small dataset}\label{section:experiments_results}

We implemented our heuristic and our streaming approach in Flexfringe \citep{flexfringe, flexfringeRepo}. We compared with the Alergia-algorithm \citep{alergia_1994} as implemented in the framework, and with the heuristic of \cite{balle_2012}, herein after called \textit{SpaceSave}. In order to compare the results we used the PAutomaC-dataset \citep{verwer_pautomac}, consisting of 48 scenarios. Each scenario comes with a set of example traces extracted from existing automata, and the task is to infer the original automaton. Each scenario also comes with a test set, consisting of string-probability pairs, assigning a probability of each string to the real automaton. In the subsequent experiments we first learn a state machine for each scenario from the dataset. We can test the quality of our learned state machines via comparison of the divergences in between the assigned probability-distributions by the learned machine and the provided true probability of the strings to occur. Similar to the PAutomaC-competition \citep{verwer_pautomac}, we use the perplexity score to compare the two distributions. The smaller the perplexity, the closer the two distributions are. All results that we report are produced by a notebook with the following relevant specifications: 

\begin{enumerate}
    \item Operating system: Ubuntu 20.4
    \item CPU: Intel i7@2.60Ghz
    \item RAM: 16GB
\end{enumerate}

%We use the perplexity \equationref{eq:perplexity} of two probability measures $p(x)$ and $q(x)$ over the same $\sigma$-algebra $\mathit{A}$ and $x\in \mathit{A}$. The smaller the perplexity, the closer the two distributions are.

%\begin{equation}\label{eq:perplexity}
%    Perp(p, q) = 2^{-\sum_{x \in \mathit{A}} p(x)log(q(x))}
%\end{equation}

\subsection{The heuristic}

To compare the heuristics we first run them all in normal \textit{batch mode}. In order to better compare Alergia with the other two heuristics we augment Alergia with the well-known \textit{k-tails}-algorithm \citep{ktails}. In this case the merge procedure remains the same as in the Alergia algorithm, but the consistency-check includes checking the sub-trees of each node pair up to a depth of $k$. The reason for this is that the Alergia heuristic naturally does not have the lookahead-feature that CSS provides. In our experiments we found that the results stop improving after $k=3$, so we only report results up to $k=3$. We then run the same experiments with CSS, varying the $F_S$ parameter from a length of $1$ up to a length of $4$. Note that Alergia with $k=3$ and CSS with $F_S=4$ both look $4$ steps ahead of each state, hence the comparison is fair. We also test CSS-MinHash, where we hash down the strings to a size of $2$, while keeping $F_S$ at $3$ and $4$ respectively. Last but not least we did the same experiments with the SpaceSave heuristic. We report the perplexities of all heuristics but the SpaceSave heuristic in Fig. \ref{fig:batch_vs_stream_box_plots}, and we report the results of the SpaceSave heuristic in Fig. \ref{fig:boxplots_space_save}, including a slightly changed version we will discuss in Section \ref{sec:space_save_discussion}. We chose two plots because the perplexities for the SpaceSave are much higher, hence putting them into the same plot as the other heuristics would make the results of the other heuristics illegible. 

Additionally, for all the steps we measured the times it took to run through. The times are in \tableref{tab:batch_runtimes}.

\begin{table}[h]
\centering
\begin{tabular}{ ||c||c|c|c|c|c|| }
 \hline
 Heuristic & $F==0$ & $F==1$ & $F==2$ & $F==3$ & $F==4$ \\
 \hline\hline
 Alergia & 1m45s & 1m55s & 2m8s & 2m24s & -  \\ 
 CSS & - & 2m1s & 2m15s & 2m14s & 2m50s  \\ 
 CSS-MinHash & - & - & - & 3m14s & 3m40s  \\
 %CSS Unconditional & - & - & 2m16s & 2m24s & 2m37s \\
 SpaceSave & - & 4m37s & 6m51s & - & - \\
 \hline
\end{tabular}
\caption{Runtime comparisons of the heuristics and different future length parameters $F$, $F$ being a placeholder for either $k$, $F_s$, or $L$. Empty fields do not exist or are not interesting to us.}
\label{tab:batch_runtimes}
\end {table}

We can see that the runtimes of CSS are small, and that CSS-MinHash took longer. The reason here is that for the dataset even small steps ahead are enough to distinguish them correctly, whereas the strength of CSS-MinHash comes out when dealing with larger strings. We also have to point out that the times for the SpaceSave heuristic are not optimal, since we do not use the proposed data structure for the sketches \citep{metwally}. This does however not influence the errors presented. We found that its performance is due to the employed consistency check. We discuss this in detail in Section \ref{sec:space_save_discussion} and do not test it further, since the nature of their sketches does not enable undoing. 

\begin{figure}%[H]%[htbp]
    \centering
    \subfigure[Comparison of all heuristics tested but the SpaceSave heuristic. All tested heuristics come with varying lookahead parameters ($k$ for Alergia, $F_S$ for CSS).]{\label{fig:boxplots_batch_all}
      \includegraphics[width=0.8\textwidth]{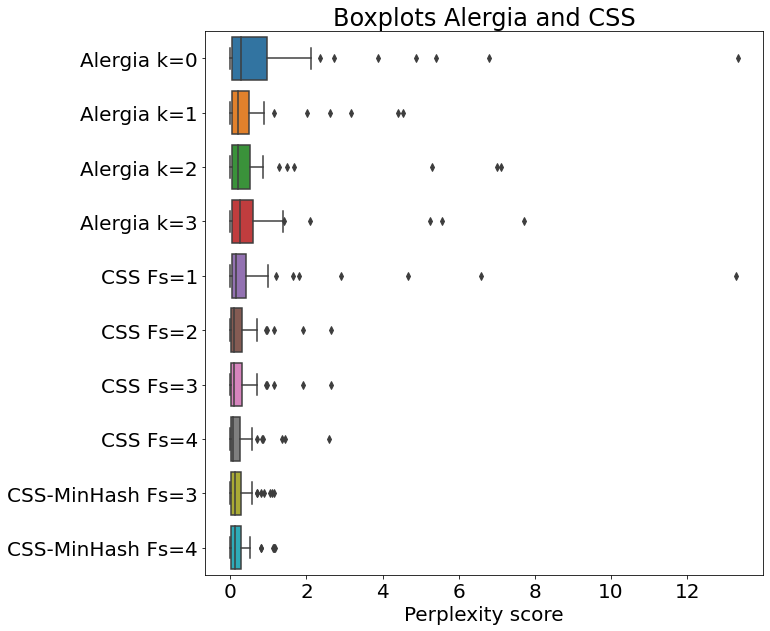}}
    \qquad
    \subfigure[Errors on SpaceSave heuristic. Pay attention to the difference in magnitude.]{\label{fig:boxplots_space_save}
      \includegraphics[width=0.73\textwidth]{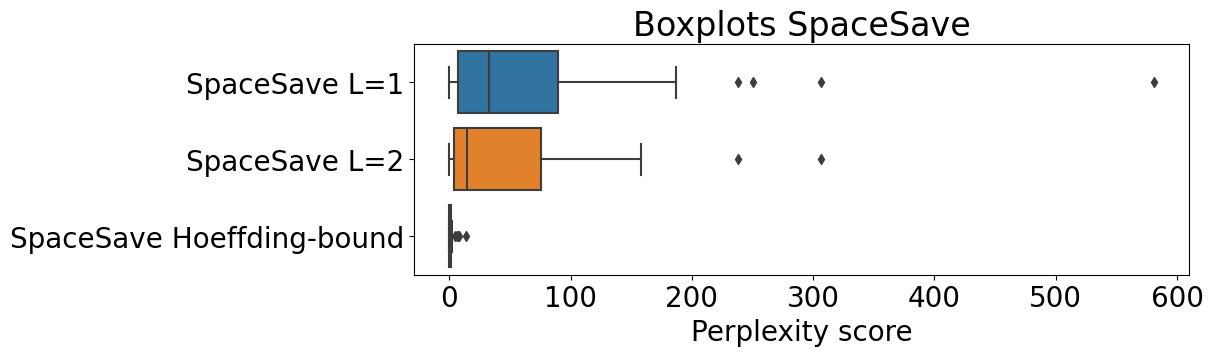}}
    \caption{Boxplots of all heuristics. Due to the difference in magitude in between the SpaceSave heuristic and the other heuristics we separated them into two sub-plots.}
    \label{fig:errors_space_save}
\end{figure}
      
\subsection{Discussion of SpaceSave}\label{sec:space_save_discussion}

As can be seen in \figureref{fig:errors_space_save}, the SpaceSave heuristic did not work all too well for us. The hyperparameters we used can be taken from \tableref{tab:params_spacesave}. While might have been able to improve via setting hyperparameters, we found it harder than the other two methods to find appropriate hyperparameters. We point this to the similarity test developed by Balle et al. and described in \citep{balle_2012}. In short, the similarity test assumes a distinguishability parameter $\mu$ for a pair of two sketches that is being approximated by value and a confidence interval estimated from sampled sketches. The term for the upper bound is determined by \equationref{eq:vc_upper_bound}. It is clear the the term includes the estimate $\hat{\mu}_k$ plus a confidence interval. The confidence interval itself is the sum of $8\nu$ (in our implementation we decreased it to $2\nu$ because it is a large constant even with relatively small $\epsilon$), with $\nu=2\epsilon$, a hyperparameter, and a max-term, where both terms within the max-term come in the form of an nth-square-root of $\frac{1}{M} \ln \frac{K_k}{\delta}$, with the pooled mass $M=\frac{m_1*m_2}{\left(\sqrt{m_1} + \sqrt{m_2}\right)^2}$ and $K_k \propto (m_1 + m_2)^2$. $m_1$ and $m_2$ are the sizes of the two respective sketches, i.e. \textit{the number of strings processed} \citep{balle_2012}. Comparing $M$ and $K_k$ in the limit leads to the following: $\lim_{m_1\rightarrow \infty} K_k = \infty$, but $\lim_{m_1\rightarrow \infty} M = m_2$, and vice versa with $m_2\rightarrow \infty$ due to the symmetry in $m_1$ and $m_2$. This means that for imbalanced states, i.e. one that has processed many strings, while the other has only processed few of them, the denominator becomes smaller and the nominator becomes larger within the max-term, i.e. the confidence bound becomes larger, even though both states might have gathered enough evidence to confidently conclude good approximation of their sampled distributions.

We show this problem with two tests that actually occurred in problem 5 of PAutomaC, an example where $m_1$ and $m_2$ are rather similar, the balanced example, and one imbalanced example, where the sizes of the states were rather unequal. In the balanced example we had one state with a size of approximately $m_1 \approx 8000$, and $m_2 \approx 11000$. The pooled mass became $M=4804$. As for the imbalanced example, $m_1 \approx 25000$ and $m_2 \approx 950$, hence $M=1016$. $K_k$ became very large in both cases, so the log term $ln \frac{K_k}{\delta}$ became $\approx 42$ in both cases. In both cases, the max-term was determined by the left side of the max-term in Eq. \ref{eq:vc_upper_bound}, and the term $\frac{1}{1c_1M}$ became $\approx 0.001$ for the balanced state pair, and $\approx 0.005$ for the imbalanced state pair, resulting in a max-term of $0.2$ for the balanced states, and $0.45$ for the imbalanced states. Given that the distinguishability of two states based on the $L_{\infty}^p$-distance as defined in \citep{balle_2012} lies in the interval $[0, 1]$, this a large change. It should be noted that in both cases the approximated distinguishability $\hat{\mu}_k$ was very small, and we considered $950$ a good enough sample size.

We provided a quick fix for the problems via simply replacing the similarity check with the Hoeffding-bound that is also used in Alergia and CSS. The fix, herein after called SpaceSave Hoeffding-bound, is depicted in Fig. \ref{fig:boxplots_space_save} along with the original results. As can be seen this improved the performance greatly, as well as runtime, because it makes the bootstrapping approach obsolete, simply relying on a single sketch per state. Comparing the results individually it is comparable to our earlier experiments run by CSS and Alergia as well. However, because the nature of the sketches does not permit the undo operation, which is required for our merging strategy, we did not proceed with this heuristic from here on. It should be noted that the large runtimes are partly attributable to the data structure we used to model the sketches described in the original paper, since we do not use the data structure described by \citep{metwally}.  %, as we show in \figureref{fig:space_save_with_hoeffding}, in which we compare the heuristic vs Alergia with $k==2$. A better performance might be achieved with some tuning, which we did not do at this point. Because the nature of the sketches does not permit the undo operation, which is required for our merging strategy, we did not proceed with this heuristic from here on. It should be noted that the large runtimes are partly attributable to the data structure we used to model the sketches described in the original paper, since we do not use the data structure described by \citep{metwally}. 

\begin{equation}\label{eq:vc_upper_bound}
    \mu^* \leq \min_{1\leq k \leq r^2} \left \{ \hat{\mu}_k + 8\nu + max\left  \{ \sqrt{\frac{1}{2c_1M} \ln \frac{K_k}{\delta}}, \sqrt[4]{\frac{(16\nu)^2}{2c_2M} \ln \frac{K_k}{\delta} } \right\} \right \}.
\end{equation}

\begin{table}
%\begin{center}
\centering
\begin{tabular}{ ||c|| c | c || }
 \hline
 Parameter & $L==1$ & $L==2$ \\
 \hline\hline
 $\epsilon$ & $0.005$ & $0.005$ \\ 
 $\delta$ & $0.05$ & $0.05$ \\ 
 $\mu$ & $0.6$ & $0.4$ \\ 
 $K$ (number of prefixes to track) & $20$ & $20$ \\ 
 $r$ (number of bootstrapped sketches) & $10$ & $10$ \\ 
 \hline
\end{tabular}
\caption{Hyperparameters used for the SpaceSave heuristic. The parameter's meaning can be taken from \citep{balle_2012}.}
\label{tab:params_spacesave}
%\end{center}
\end {table}

% \begin{figure}[!ht]
%     \centering
%     \includegraphics[width=0.8\textwidth]{Figures/perplexities_spacesavehoeffding_vs_alergia.jpg}
%     \caption{New results of the SpaceSave heuristic with the Hoeffding bound instead of the old check.}
%     \label{fig:space_save_with_hoeffding}
% \end{figure}

\subsection{Streaming}

After measuring the base-performance of our Count-Min-Sketch heuristic, we are also interested in the effect of our new streaming method. Having already experimented with different parameters, we decided to move on with Alergia $k=3$ and CSS-MinHash $F_s=4$, because these delivered the best results respectively. Note again that both look four steps into the future per state. To compare our streaming approach we compare it with the traditional streaming approach that does not undo and redo refinements, which we call \textit{streaming old}. We depict the perplexities of the streaming approaches in Fig. \ref{fig:batch_vs_stream_box_plots}. Because we also want to see how the streaming compares with the batch-mode version, we picked the best streaming version, i.e. the CSS-MinHash with $F_S=4$, and compared it with its respective batch-mode version in Fig. \ref{fig:css_batch_vs_stream}.

\begin{figure}%[H]%[htbp]
    \centering
    \subfigure[Perplexity scores of the stream settings and the new stream setting, represented in boxplots.]{\label{fig:batch_vs_stream_box_plots}
      \includegraphics[width=\textwidth]{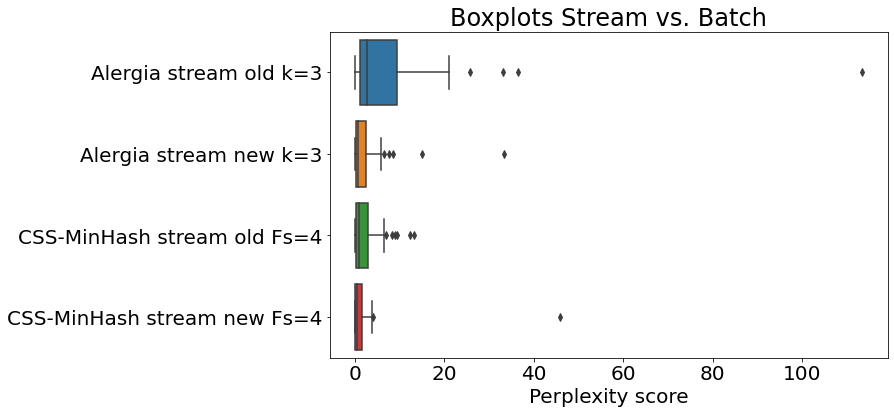}
    }
    \qquad
    \subfigure[Perplexity scores CSS-MinHash in batch-mode vs. the new stream-mode. Problem $20$ the stream mode had a perplexity score of around $40$.]{\label{fig:css_batch_vs_stream}
      \includegraphics[width=\textwidth]{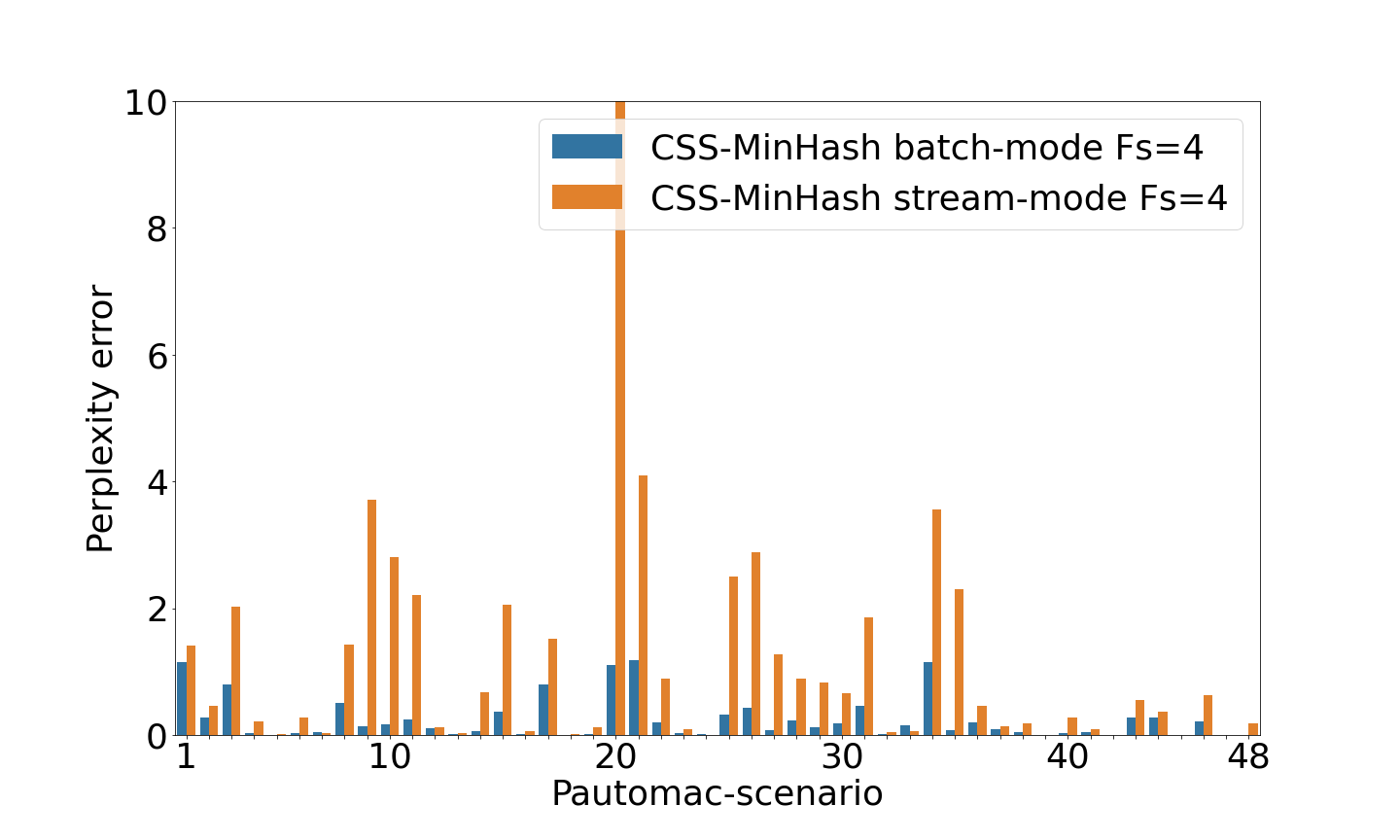}}
    %\caption{Perplexity score results of the streaming experiments.}
    \label{fig:box_plots_batch_mode}
\end{figure}

The run-times are as follows: Alergia in the old stream setting took 37s, while CSS took 67s, and on the new stream setting they took 51s (Alergia) and 78s (CSS). We report the maximum achieved heap size as measured using the Valgrind tool \citep{valgrind} in \tableref{tab:memory_css}. As we can see the streaming did drastically reduce the memory footprint and run-time compared with the batch version. We can also see that the old streaming approach is faster but gives worse results, and that the errors for Alergia are much higher on the old streaming approach, since on the new streaming approach it gets the chance to correct its mistakes. The main advantage of CSS over the k-tails augmented Alergia is then that CSS has more information on the fringes of the prefix tree. We see that in the new streaming CSS still performs better than Alergia on most problems, although the difference has become much smaller than on the old streaming setting. 

\begin{table}
\centering
\begin{tabular}{ ||c||c|c|c|c|c|c|c|| }
 \hline
 Strategy & Prob. 1 & Prob. 5 & Prob. 11 & Prob. 18 & Prob. 23 & Prob. 34 & Prob. 45 \\
 \hline\hline
 Batch & 250MB & 56MB & 1.1GB & 3.1GB & 1.0GB & 2.4GB & 257MB \\ 
 Stream old & 15MB & 13MB & 34MB & 123MB & 63MB & 53MB & 13MB \\ 
 Stream new & 14MB & 13MB & 29MB & 123MB & 63MB & 44MB & 13MB \\
 \hline
\end{tabular}
\caption{Memory-footprints compared on a few selected problems. Heuristic employed: CSS-MinHash $F_s==4$ .}
\label{tab:memory_css}
\end {table}

\section{Analysis}\label{sec:analysis}

\subsection{PAC learning}

%TODO: very similar to Balle's chapter. Should I change it more to look less like his?

The \underline{P}robably \underline{A}pproximately \underline{C}orrect (PAC) learning framework is a framework to analyze machine learning algorithms. Formally, we are given a class of distributions $\mathcal{D}$ over some set $X$, where $X$ remains fixed. A machine learning algorithm then PAC-learns $\mathcal{D}$ if for each $\delta>0$ and each $\epsilon > 0$ the algorithm takes $S(\frac{1}{\epsilon}, \frac{1}{\delta}, |D|)$ i.i.d. drawn samples from a $D \in \mathcal{D}$ and returns a hypothesis $H$ s.t. $L_1(D, H) \leq \epsilon$. In this notation, $|D|>0$ is a measure of complexity over $\mathcal{D}$. We say a PAC learner is efficient if $S( \cdot )$ and $T( \cdot )$ are polynomial in all their respective parameters.

In order to show PAC-bounds on our algorithm we need to analyze a few components of the algorithm. We will start from the ground up: Since the CMS do not represent the exact distribution, but instead only return approximations of distributions, which are sampled themselves, we first need to check how the CMS influence the sampled distributions we feed into the statistical test. Secondly, we need to error bound the statistical test. And thirdly, we need to analyze how the data stream constructs initial states, and bound the error of the actual graph creation through state merging.

\subsection{Analysis of Count-Min-Sketches}

To give error bounds on the sketches we first need a few definitions and notations from \cite{balle_2012}. Let $*$ denote the Kleene-operation. Let furthermore $D_1$ and $D_2$ be distributions on $\Sigma^*$, the set of all finite prefixes over $\Sigma$ and the empty string. Let $D_i(x\Sigma^*)$ be the probability of having prefix $x$ under distribution $D_i$, and let $L_{\infty}^p(D_1, D_2) = \max_{x\in \Sigma^*} \left| D_1(x\Sigma^*) - D_2(x\Sigma ^*)\right|$ be the prefix-$L_{\infty}$ distance between $D_1$ and $D_2$. We denote by $|\Sigma|$ the cardinality of the alphabet, and for ease of notation we assume the final symbol $\xi$ to be part of the alphabet from now on: $\Sigma \leftarrow \Sigma \cup \xi$. Finally we do need a formal definition of to distinguish two distributions, something we will extensively use throughout the remainder of this section. Following \cite{clark_pac_paper} and \cite{palmer_pac_paper} we define the distinguishability of two distributions and states.

\vspace{\baselineskip}

\begin{customdefinition}[$\mu$-Distinguishability]
    We call two distributions $D_1$ and $D_2$ $\mu$-distinguishable when $L_{\infty}^p(D_1, D_2) \geq \mu$. We call a PDFA $T$ $\mu$-distinguishable when $\forall q_1, q_2\in T: L_{\infty}^p\left(D_{q1}, D_{q2}\right) \geq \mu$. 
\end{customdefinition}

\noindent In a first, we want to look at the error that the CMS introduce on the Alergia-test of the original sampled distribution. Since the CMS approximate the sampled distributions, we can either overestimate or underestimate the left hand side of Eq. \ref{eq:hoeffding}, the statistical test we perform and first mentioned in Alg. \ref{alg:consistency}. The test compares a pair of two states and tests them on similarity, either revealing similar or dissimilar, corresponding with the binary outcome of wether the two distributions are deemed equal or not.

\begin{equation}\label{eq:hoeffding}
     \left| \frac{x_i^l}{n_l} - \frac{x_i^r}{n_r} \right| < \sqrt{\frac{1}{2}log\left(\frac{2}{\alpha}\right)}\left(\frac{1}{\sqrt{n_l}}+\frac{1}{\sqrt{n_r}}\right),
\end{equation}

\noindent Formally, the test is derived from the Hoeffding-bound (see \cite{alergia_1994}), which is why we also refer to it as the Hoeffding-test. In this equation, the left-hand-side represents two sampled distributions of elements $x_i$, and $n_l$ and $n_r$ are the respective sample sizes. Hence, $\frac{x_i^l}{n_l}$ represents the frequency of sampled element $x_i$ in the left distribution, and $\frac{x_i^r}{n_r}$ its sampled frequency in the right distribution. The two states are deemed similar if the inequality holds for all elements $x_i$ that are in the two sampled distributions. We state the following theorem: 

\begin{customtheorem}\label{theorem:upper_bound_hoeffding}
    Given the parameters ($\beta$, $\gamma$)\footnote{By convention these are be called $\epsilon$ instead of $\beta$ and $\delta$ instead of $\gamma$, however we already used these two letters.}, and given two CMS whose width $w$ and depth $d$ of a CMS assume $w=\left\lceil{\frac{e}{\beta}}\right\rceil$ and $d=\left\lceil{\ln \frac{1}{\gamma}}\right\rceil$. The approximation error on the left hand side of the Hoeffding-test (Eq.~\ref{eq:hoeffding}) is upper-bounded by $\beta$ with probability at least $1-2\gamma$.
\end{customtheorem}

\begin{proof}
The following definitions are from~\citep{count-min-sketch}. We need indicator variables 

\begin{equation}
    I_{i,j,k}=
    \begin{cases}
    1 , & \text{if } h_k(i) == h_k(j) \text{ and } i\not=j\\
    0,              & \text{otherwise},
\end{cases}
\end{equation}

\noindent again with $h_k(i)$ being the hash-function associated with row $k$, hashing element $i$, and we need a term for the collisions:

\begin{equation}
    C_{i, k} = \sum_{j=1}^w I_{i,j,k}x_j.
\end{equation}

\noindent Assuming the hash functions hash uniformly we further get the following two inequalities from~\citep{count-min-sketch}: 

\begin{equation}
   \mathbb{E}[I_{i,j,k}] = P\left(h_k(i)=h_k(j)\right) = \frac{1}{w} \leq \frac{\beta}{e}, 
\end{equation}

and 

\begin{equation}
    \mathbb{E}[C_{i,k}]=\mathbb{E}\left[\sum_{j=1}^w I_{i,j,k}x_j \right] \leq \sum_{j=1}^w x_j\mathbb{E}[I_{i,j,k}] \leq m\frac{\beta}{e},
\end{equation}

\noindent with $m$ being the total counts per row of the CMS, and $e$ being Euler's constant. With these definitions, we can now bound the error introduced by the CMS. As by the nature of the algorithm, we know that all counts $x_{i,j}$ in the CMS are greater than or equal to zero at all times. We define a new random variable $Z_i=\left| \frac{1}{n_1}\min_k count(x_{i,k}) - \frac{1}{n_2}\min_k count(y_{i, k})\right|$, where $x_{i,k}$ are the entries in the first CMS, and $y_{i,k}$ the entries in the second CMS. We consider two cases: 1. The non-absolute (inner) term of $Z_i$ is $\geq 0$,  and 2. the non-absolute (inner) term of $Z_i$  is $< 0$. Taking the first case and assuming our sampled distributions are $\mu$-distinguishable\footnote{Alternatively we could use any other placeholder constant for the real difference in distributions, since we are only interested in the error.} we get a  $k_1$ and a $k_2$ such that the minimum term is satisfied and hence

\begin{multline}\label{eq:Z_i_basic}
%\begin{center}
    Z_i = \frac{x_{i,k_1}}{n_1} - \frac{y_{i,k_2}}{n_2} + \frac{1}{n_1}\sum_{j=1}^w I_{i,j,k_1}x_{j,k_1} - \frac{1}{n_2}\sum_{j=1}^w I_{i,j,k_2}y_{j,k_2} \leq \\ \mu + \frac{1}{n_1}\sum_{j=1}^w I_{i,j,k_1}x_{j,k_1} - \frac{1}{n_2}\sum_{j=1}^w I_{i,j,k_2}y_{j,k_2} \leq \mu + \frac{1}{n_1}\sum_{j=1}^w I_{i,j,k_1}x_{j,k_1}.
%\end{center}
\end{multline}

\noindent The first inequality comes from $\mu$ as the upper bound on the real distributions, and the second one from the fact that our entries in the CMS are all zero or positive. Following the above, we do get the following bound on the expected value:

\begin{equation}
    \mathbb{E}[Z_i] \leq \mu + \mathbb{E}\left[\frac{1}{n_1}\sum_{j=1}^w I_{i,j,k_1}x_{j,k_1} - \frac{1}{n_2}\sum_{j=1}^w I_{i,j,k_2}y_{j,k_2}\right] \leq \mu + \mathbb{E}\left[\frac{1}{n_1}\sum_{j=1}^w I_{i,j,k_1}x_{j,k_1}\right] \leq \mu + \frac{\beta}{e}.
\end{equation}

Finally using the Markov-inequality we get 

\begin{multline}
    P(Z_i\geq \mu + \beta) \leq P\left(\mu + \min_k \frac{1}{n_1}\sum_{j=1}^w I_{i,j,k}x_{j,k} \geq \mu + \beta \right) = P\left(\min_k \frac{1}{n_1}\sum_{j=1}^w I_{i,j,k}x_{j,k} \geq \beta \right) = \\ P\left(\forall k: \frac{1}{n_1}C_{i,k} \geq \frac{e}{n_1}\mathbb{E}[C_{i,k}]\right)=\frac{1}{e^d} = \gamma. 
\end{multline}

\noindent The other direction, when $Z_i<0$, can be shown analogously. We end up with two cases: Either the first term gets overestimated, or the second term gets underestimated, and a union bound brings us to $P(|Z_i| < \mu + \beta) \geq 1-2\gamma$.

\end{proof}

\noindent This settles the upper bound for the left hand side of the test. But we are also interested in the lower bound, given by the following theorem:

\begin{customtheorem}\label{theorem:hoeffing_test_lower_bound}
    With a probability of at least $1-2\gamma$ the left hand side of the Hoeffding-test with CMS is not underestimated by an offset larger than $\beta$. 
\end{customtheorem}

\begin{proof}
    Taking $Z_i$, we want a lower estimate this time. Following the proof of Theorem~\ref{theorem:upper_bound_hoeffding}, we can get a wrong estimates via two causes in Eq.~\ref{eq:Z_i_basic}: Either the second sum grows in relation to the first, or the first sum shrinks in relation to the second. We do get an upper estimate of the error by holding one of the two sums fix and increasing the other. Take the case that the first sum is larger than the second, then fixing the first some to $\psi_1 = \frac{1}{n_1}\sum_{j=1}^w I_{i,j,k_1}x_{j,k_1}$

    \begin{equation}
        Z_i = \mu + \psi_1 - \frac{1}{n_2}\sum_{j=1}^w I_{i,j,k_2}y_{j,k_2}.
    \end{equation}

    \noindent In this case, 

    \begin{equation}
        P \left( Z_i \leq \mu - \beta \right)=P\left(\forall k_2: \mu - \frac{1}{n_2}\sum_{j=1}^w I_{i,j,k_2}y_{j,k_2} \geq \mu - \beta \right) \leq \gamma.
    \end{equation}

    \noindent The proof to the last inequality can be found in Theorem 1 of~\citep{count-min-sketch}. The alternative case is when the second term is larger than the first one, in which case we get an upper estimate via fixing the second term through $\psi_2 = \frac{1}{n_2}\sum_{j=1}^w I_{i,j,k_2}y_{j,k_2}$, and 

    \begin{equation}
        Z_i = \mu + \frac{1}{n_1}\sum_{j=1}^w I_{i,j,k_1}x_{j,k_1} - \psi_2.
    \end{equation}

    \noindent This case is analogue to the first one, except that we have to take the absolute value in the Hoeffding-test into account, i.e.

    \begin{equation}
        P \left(|Z_i| \leq \mu - \beta \right)=P \left( \forall k_1: \mu - \frac{1}{n_1}\sum_{j=1}^w I_{i,j,k_1}x_{j,k_1} \geq \mu - \beta \right) \leq \gamma.
    \end{equation}

    \noindent Taking a union bound over the two cases leads us to $P(|Z_i| \leq \mu - \beta)\geq 2\gamma$, and hence. 

    \begin{equation}
        P \left( |Z_i| \geq \mu - \beta \right)\geq 1- 2\gamma.
    \end{equation}
    
\end{proof}

\noindent The following are a few statements on the heuristic that will be useful later.

\begin{customcorollary}[Run-time of heuristic]
    Processing strings takes time $\mathcal{O}\left(F_s \right)$, and performing the Hoeffding-test takes time $\mathcal{O}\left(|\Sigma|^{F_s}\right)$.
\end{customcorollary}

\begin{proof}
    From~\citep{count-min-sketch} we already know that both a store and a retrieve from a CMS take time $\mathcal{O}\left(1 \right)$. The bottleneck when storing is the number of futures, and when retrieving distributions the bottleneck is the number of elements we can get.
\end{proof}

\noindent The CSS-conditional framework greatly reduces run-time, but of course we do not model conditional probabilities anymore and hence break some of the assumptions on the distributions.

\begin{customlemma}
    Under CSS-MinHash the run-time for the Hoeffding-test reduces to $\mathcal{O}\left(|\Sigma|^{l_{m}}\right)$.
\end{customlemma}

%We also want to describe the error introduced by the MinHash-algorithm. 

%\begin{customcorollary}[Compression rate of MinHash]
%    Under the CSS-MinHash scheme, for n-grams whole length $l$ is larger than $l_{m}$ the compression rate under maximum entropy is $\frac{|\Sigma|^{l_{m}}}{\sum_{i=l_{m} + 1}^{F_s} |\Sigma|^i}$. TODO: fix this corrollary
%\end{customcorollary}

\subsection{Bounding the statistical test}

In the following we want to bound the actual errors that can happen through the statistical test. A binary test can make two kinds of errors: A false positive (false accept), or a false negative (false reject). To bound those two we need an extra parameter, namely the expected number of states $n$. We also have to assume that the strings $\sigma$ we draw from $D$ are i.i.d. distributed and that the distribution $D$ remains fixed in time. Obviously our confidence in the statistical test will grow with more samples, therefore we want an intuition of how data points propagate through the unmerged automaton:

\begin{customtheorem}[Samples processed]
    Given batch-size $B$, upper bound on state $n$, and $p_{min}$ the minimum transition probability $\lambda_{min}=\min_{i, j} \lambda(q_i, a_j)$. The algorithm reads $\mathcal{O}(nBN_{max})$ samples from the stream, where $\mathbb{E}[|\lambda(q', a')|] = N_{max} \cdot \lambda_{min}$, with $q', a'$ chosen s.t. $\lambda_{min}=\lambda(q', a')$.
\end{customtheorem}

\begin{proof}
    We define the non-trivial set of prefixes $\mathcal{F}$ the set of all prefixes whose probability is not zero under $D$: $\mathcal{F}=\left\{\sigma\in \Sigma^* | D(\sigma)>0 \right\}$. Further, let $m_0$ be the threshold before a state can be promoted blue, and $B$ the batch-size. Consider only the root-node $q_0$. In order to create a blue state $q_1$ with $\tau(q_0, a_1, q_1)$ we need at least $m_0$ strings of the form $a_1\Sigma^*$. Then, by Markov's inequality for some $N \leq N_{max}$

    \begin{equation}
        P(|q_1|\geq m_0) \leq \frac{ND(a_1\Sigma^*)}{m_0} \leq \frac{N_{max}D(a_1\Sigma^*)}{m_0} \leq\frac{N_{max}}{m_0},
    \end{equation}

    \noindent i.e. the probability grows $\mathcal{O}(N)$. Denote the number of samples once $q_1$ has been created by $N_{a_1}$. Promoting a new state with $\tau(q_1, \sigma_2, q_2)$ is then simply the sum of $N_{a_1}$ and $N_{a_2}$, i.e. the sum of linear growth. Similarly, the batch-size $B$ adds a constant of maximum $B-1$ strings in, and hence we get $\mathcal{O}(nBN)$. The rest of the proof follows by induction.
\end{proof}

\noindent We further want to know how many samples to read in order to bound the probability of an error. The bound on a false reject is a mathematical component of the test itself, the sample size will become more relevant when bounding false accepts. False rejects are bounded by the following theorem:

\begin{customtheorem}
    Given the Hoeffding-test as defined earlier, and assuming no collisions happened in the CMS, then the probability  of a false reject is kept below $2\alpha$.
\end{customtheorem}

\begin{proof}
    The proof can be extracted from~\citep{alergia_1994}. Given the original Hoeffing-bound and analyzing the error we get for random variable $X$ and sampled distribution $\frac{x}{m}$ the inequality

    \begin{equation}
        P\left(\left| \frac{x}{m}-\mathbb{E}[X] \right|\geq \sqrt{\frac{2}{m}\ln{\frac{2}{\alpha}}}\right) \leq \alpha.
    \end{equation}

    \noindent Hence, assuming that two sampled distributions $\frac{x}{m_1}$ and $\frac{y}{m_2}$ are equal in expectation the error of the Hoeffding-test is kept below $2\alpha$. 
\end{proof}

\noindent The false accept we will do in two steps: The following theorem again starts with the assumption that no collision inside the CMS happens, and we will expand on it allowing for collisions after that.

\begin{customtheorem}\label{theorem:false_reject_bound_no_collision}
    Assuming the target PDFA we are trying to learn is $\mu$-distinguishable, there exists a minimum sample size $m_0$ per state such that the probability of a false accept per state-pair is bounded by $\frac{\delta'}{2n^2|\Sigma|^2}$ under the assumption of no collisions.
\end{customtheorem}

\begin{proof}
    At first we need to bound the variance of $C_{i,j}$. We can show that 

    \begin{equation}
        Var(C_{i,k}) = \sum_{j=1}^w x_j^2 Var(I_{i,j,k}) + \sum_{i\neq j} x_i x_j Var(I_{i,j,k}) \leq w\frac{1}{w}\left(1-\frac{1}{w}\right)\left(\sum_{j=1}^w x_j^2 + \sum_{i\neq j} x_i x_j\right).
    \end{equation}

    \noindent In this equation we implicitly turned a covariance into a variance, since the random variable is the same, and the variance term resolved to the variance of the binomial distribution with probability of a hit of $p=\frac{1}{w}$. Let us denote the random variable representing the frequency via $C_{i,k}^\nu=\frac{C_{i,k}}{m}$, with $m$ the number of samples in the sampled distribution. Then

    \begin{equation}
        Var(C_{i,k}^\nu) \leq \left(1-\frac{1}{w}\right)\left(\sum_{j=1}^w \frac{x_j^2}{m^2} + \sum_{i\neq j} \frac{x_i x_j}{m^2}\right) \leq 2\left(1-\frac{1}{w}\right).
    \end{equation}

    \noindent The last inequality comes from $\sum_{i}x_i^2 \leq \left(\sum_{i}x_i\right)^2$ and the Cauchy-Schwarz inequality. Taking the left hand side of the Hoeffing-bound $X_{i,k}^\nu - Y_{i,k}^\nu$, both constructed like $C_{i,k}^\nu$ but for two different distributions $D_1$ and $D_2$, we get

    \begin{equation}
        Var(X_{i,k_1}^\nu - Y_{i,k_2}^\nu) = Var(X_{i,k_1}^\nu) + Y_{i,k_2}^\nu - Cov(X_{i,k_1}^\nu, Y_{i,k_2}^\nu) \leq Var(X_{i,k}^\nu) + Y_{i,k}^\nu \leq 4\left(1-\frac{1}{w}\right).
    \end{equation}

    \noindent The covariance term cancels out because we require that the hash functions be mutually independent, i.e. $Cov(X_{i,k_1}^\nu, Y_{i,k_2}^\nu) = 0 \text{ }\forall k_1 \neq k_2$, and $Cov(X_{i,k_1}^\nu, Y_{i,k_2}^\nu) > 0 \text{ iff } k_1 = k_2$. Therefore, $Cov(X_{i,k_1}^\nu, Y_{i,k_2}^\nu) \geq 0 \text{ } \forall i, k_1, k_2$. Additionally, we do get the following upper bound for the variance of the Hoeffing-test from~\citep{alergia_1994}:

    \begin{equation}\label{eq:variance_bound_alergia}
        Var(f_1 - f_2) \leq \frac{1}{4m_1} + \frac{1}{4m_2} \leq \frac{1}{2m_0}.
    \end{equation}

    \noindent We can easily verify that the latter bound from Eq.~\ref{eq:variance_bound_alergia} is the sharper bound, even for a sample size of $m_0=1$ (a state with size $m=0$ cannot exist in our framework by design). In order to have a false accept, we do have that the two distributions are unequal by at least $\mu$, yet we do satisfy the bounds. Hence, we need to have the scenario that the true distribution $|D_1(i) - D_2(i)| = \mu$, but the expression $|\delta f^\nu| = \left|\min_{k_1}X_{i,k_1}^\nu - \min_{k_2}Y_{i,k_2}^\nu\right| < \sqrt{2\ln\left(\frac{2}{\alpha}\right)}\left(\frac{1}{\sqrt{m_1}} + \frac{1}{\sqrt{m_2}}\right)$. From the Chebychev-Cantelli inequality for lower bounds $P(X - \mathbb{E}[X] \leq -\lambda) \leq \frac{Var(X)}{Var(X) + \lambda^2}$ and the fact that $\mu$ represents a lower bound on distinguishability we get

    \begin{equation}
        P\left(|\delta f^{\nu}| - \mu \leq -\left(\mu - \sqrt{\frac{2}{m_0}\ln\left(\frac{2}{\alpha}\right)}\right)\right) \leq \frac{\frac{1}{2m_0}}{\frac{1}{2m_0} + \left(\mu - \sqrt{\frac{2}{m_0}\ln\left(\frac{2}{\alpha}\right)}\right)^2}.
    \end{equation}

    \noindent We denote the term $f(m_0)=\frac{\frac{1}{2m_0}}{\frac{1}{2m_0} + \left(\mu - \sqrt{\frac{2}{m_0}\ln\left(\frac{2}{\alpha}\right)}\right)^2}$ and require $f(m_0)\leq \frac{\delta'}{2n^2|\Sigma|^2 (wd)^2 F_s}$. It is easy to see that $\lim_{m_0\rightarrow \infty}f(m_0)\rightarrow 0$, hence $\exists m_0 \text{ s.t. } \forall \delta'>0: f(m_0)\leq \frac{\delta'}{2n^2|\Sigma|^2 (wd)^2 F_s}$. The final statement then follows from the fact that we have a maximum of $(wd)^2$ possible cell-pairs to check on $F_s$ sketches, and a union bound brings us to to the probability of a false accept to $\frac{\delta'}{2n^2|\Sigma|^2}$.
    
\end{proof}

\noindent Analyzing the term $f(m_0)$ we find that it has a maximum and then converges to zero. We plot the function in Fig.~\ref{fig:f_m_zero} for different values of $\mu$. In order to refer to the desired $m_0$ from here on we denote $m_0 = \min_m \{m\in \mathbb{N} \text{ }| \text{ }f(m') \leq \frac{\delta'}{2n^2|\Sigma|^2 (wd)^2 F_s} \text{ } \forall m' \geq m \}$. However, convergence to zero is also mathematically trivial to show, assuming that $\mu>0$, an assumption which is inherently always true to the problem (it does not make sense for $\mu$ to be zero or negative).

\vspace{\baselineskip}

\noindent Now that we know we can bound the errors when no collisions happen, i.e. when the CMS return the true counts, we have to look at what happens when the counts that the CMS return are approximates.

\begin{customtheorem}
    Given the Hoeffding-test and given two frequencies estimated via CMS, allowing for collisions. Then the probability of a false accept under an updated Hoeffding-bound $\sqrt{\frac{1}{2}log\left(\frac{2}{\alpha}\right)}\left(\frac{1}{\sqrt{m_1}}+\frac{1}{\sqrt{m_2}}\right) - \beta$ is bounded by $\frac{\delta_2'}{2n^2|\Sigma|^2}$.
\end{customtheorem}

\begin{proof}
     Following the proofs of Theorem~\ref{theorem:hoeffing_test_lower_bound} and Theorem~\ref{theorem:false_reject_bound_no_collision} the proof is trivial. We bounded the left-hand-side of the Hoeffding-test to $|\delta f^\nu| \geq \mu - \beta$ by at max $2\gamma$, and inserting the updated bound and, by independence, updating $\delta_2' = 2\gamma\delta'$ into the proof of Theorem~\ref{theorem:false_reject_bound_no_collision} gives us the result.
\end{proof}

\begin{customtheorem}
    Narrowing the Hoeffding-test via $\sqrt{\frac{1}{2}log\left(\frac{2}{\alpha}\right)}\left(\frac{1}{\sqrt{m_1}}+\frac{1}{\sqrt{m_2}}\right) - \beta$, the upper bound for a false reject remains bounded by a maximum of $2\alpha$.
\end{customtheorem}

\begin{proof}
    A false accept happens when the true distribution is larger, but the sampled distribution is smaller than the right-hand-side of the Hoeffding-test. A smaller right-hand-side cannot increase the probability of a false accept.
\end{proof}

\noindent In our algorithm we can control $\beta$ and $\gamma$ directly via choosing appropriate dimensions for the CMS. It is desirable and necessary for the guarantees to hold that $\beta < \mu$\footnote{From a mathematical perspective, because elsewise we would be dealing with zero or negative probabilities.}, which sets a lower bound on the width of the sketches. The depth directly influences the probability of mistakes made, as shown above. Analogue to above we want $\alpha$ to satisfy $\alpha \leq \frac{\delta_2'}{2n^2|\Sigma|^2 (wd)^2 F_s}$ in order to give the formal bounds below.

\begin{figure}[!ht]
    \centering
    \includegraphics[width=0.8\textwidth]{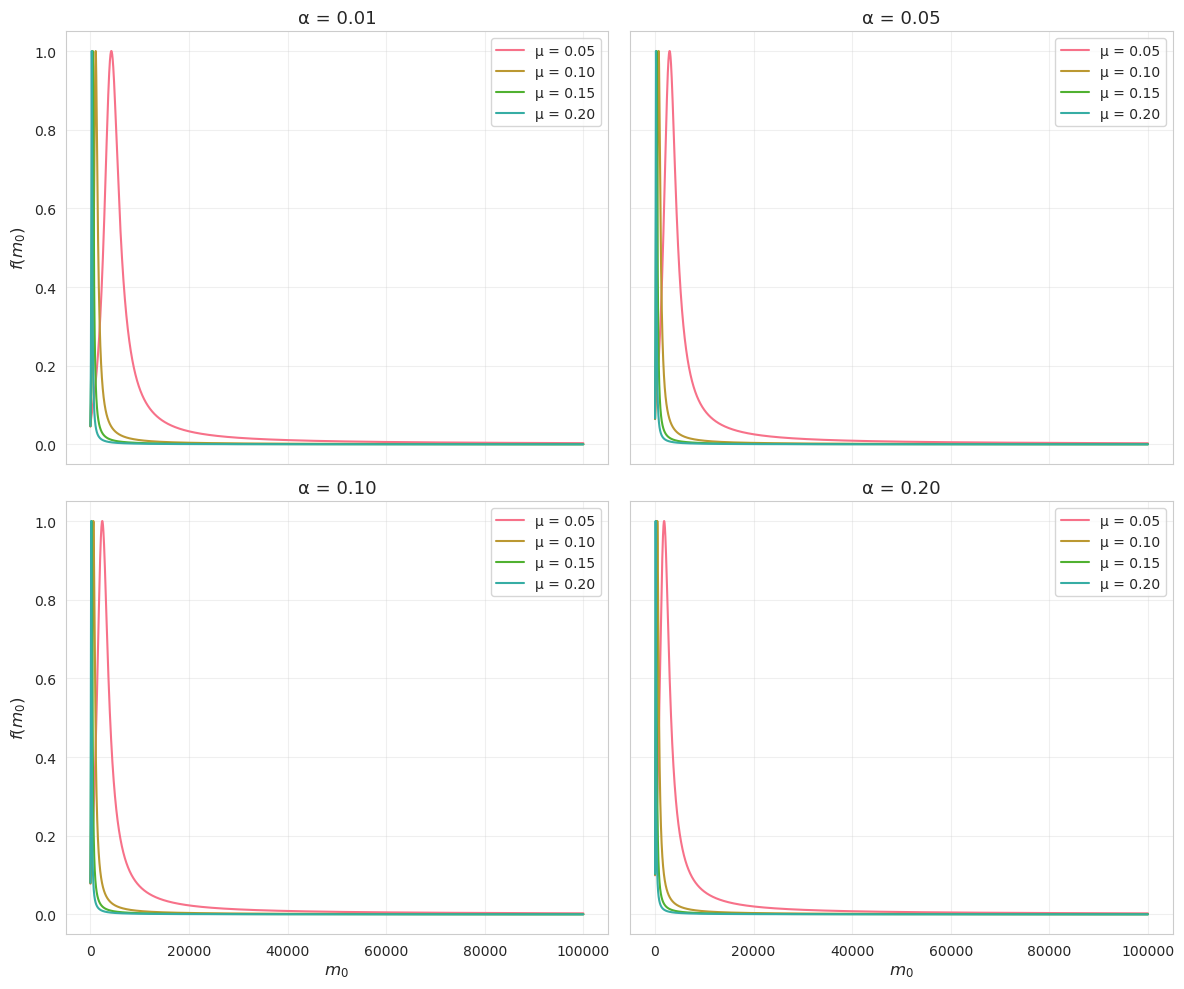}
    \caption{The function $f(m_0)$ with varying $\mu$ and $\alpha$.}
    \label{fig:f_m_zero}
\end{figure}

\subsection{Bounds on graph creation}

Now that we can bound the statistical test using the CMS we need to investigate what happens when creating the graph from the set of individual states. In the next step we do need a bound on the batch size. We follow Palmer et al. and merely repeat their Proposition 1~\citep{palmer_pac_paper} in our next theorem, which gives us a lower bound on the batch size in order to probabilistically guarantee the existence of desired states:

\begin{customtheorem}
    Let $A'$ be a PDFA whose states and transitions are a subset of those of $A$. $q_0$ is a state of $A'$. Suppose $q$ is a state of $A'$ but $\tau(q, \sigma)$ is not a state of $A'$. Let $S$ be an i.i.d. drawn sample from $D_A$, $|S|\geq \frac{8n^2|\Sigma|^2}{\epsilon^2}\ln \left( \frac{2n^2|\Sigma|^2}{\delta'} \right)$. Let $S_{q, a}(A')$ be the number of elements of $S$ of the forms $\sigma_1 a \sigma_2$, where $\tau(q_0, \sigma_1)=q$ and for all prefixes $\sigma_1'$ of $\sigma_1$: $\tau(q_0,\sigma_1')\in A$. Then

    \begin{equation}
        P\left(\left|\frac{S_{q,a}(A')}{|S|} - \mathbb{E}\left[\frac{S_{q,a}(A')}{|S|}\right]\right| \geq \frac{\epsilon}{8n|\Sigma|}\right) \leq \frac{\delta'}{2n^2|\Sigma|^2}.
    \end{equation}
\end{customtheorem}

\begin{proof}
    The proof can be found in~\citep{palmer_pac_paper}.
\end{proof}

\noindent Given the statements above, we can give formal guarantees on our algorithm by choosing the inputs $\alpha, \beta, \gamma$ appropriately. We do omit the argumentation and the proofs, as they can be extracted exactly the same from~\citep{palmer_pac_paper}, and simply arrive at the next theorem, which is a direct copy of Theorem 2 of Palmer et al.\citep{palmer_pac_paper}:

\begin{customtheorem}
    Let our algorithm read samples with a batch-size of 
    
    \[
    B \geq max\left(\frac{8n^2|\Sigma|^2}{\epsilon^2}\ln\left(\frac{2n^2|\Sigma|^2}{\delta'}\right), \frac{4m_0n|\Sigma|}{\epsilon}\right),
    \]
    
    \noindent then there exists $T'$ a subset of the transitions of $A$, and $Q'$ a subset of the states of $A$, s.t. $\sum_{(q, \sigma) \in T'} D_A(q, \sigma) + \sum_{q\in Q'}D_A(q)\leq \frac{\epsilon}{2}$, and with probability at least $1-\delta'$, every transition $(q, \sigma) \not \in T'$ in target automaton $A$ for which $D_A(q, \sigma) \geq \frac{\epsilon}{4n|\Sigma|}$, has a corresponding transition in hypothesis automaton $H$, and every state $q\not\in Q'$ in target automaton $A$ for which $D_A(q)\geq \frac{\epsilon}{4n|\Sigma|}$, has a corresponding state in hypothesis automaton $H$.
\end{customtheorem}

\noindent In other words, we are able to retrieve an isomorphic subgraph of the target automaton that holds transitions and nodes with probabilistic guarantees. From this theorem we are able to retrieve the following corollary. In the remainder of this analysis we will always assume a batch-size of at least $B\geq max\left(\frac{8n^2|\Sigma|^2}{\epsilon^2}\ln\left(\frac{2n^2|\Sigma|^2}{\delta'}\right), \frac{4m_0n|\Sigma|}{\epsilon}\right)$.

\begin{customcorollary}
    With high probability, the maximum number of iterations the algorithm performs is $n|\Sigma|$.
\end{customcorollary}

\noindent The proof can be extracted from the proof of Theorem 2 of \citep{palmer_pac_paper}. Essentially, they show that in each iteration at least one new state with a size of at least $m_0$ will appear. What is different in our case is that we can undo merges and re-compute them. The maximum worst case is however still that for each of the $n$ states we collapse the whole alphabet $|\Sigma|$ into it, i.e. the worst case does not change under these circumstances. With the maximum number of iterations we can retrieve bounds on run-time and memory consumption.

\begin{customcorollary}
    In expectation, the algorithm uses $\mathcal{O}\left( \frac{n|\Sigma|}{\beta} \left\lceil \frac{B}{m_0} \right\rceil\right)$ memory and takes $\mathcal{O}\left(n|\Sigma|B\right)$ amortized run-time. 
\end{customcorollary}

\noindent Unlike the other works we keep states after a merge in order to undo them if need be. Hence, our algorithm's memory consumption is upper bounded by the number of states the algorithm can possibly produce before its termination. The corollary then follows from the fact that the CMS use $\mathcal{O}\left(\frac{1}{\beta}\right)$ memory~\citep{count-min-sketch} for point queries. We know we have an expected maximum of $n|\Sigma|$ iterations, and hence there is a maximum of $n|\Sigma|\left\lceil \frac{B}{m_0} \right\rceil$ possibly created states in total. For the run-time we have to take into account that the algorithm can make mistakes and might have to do re-computations. From before we know that with probability $1-\delta'$ each merge is correct the first time it is executed. Undoing a merge results in a worst case in the re-computation of its entire sub-tree, which is capped by a maximum of $\left(n|\Sigma|\left\lceil \frac{B}{m_0} \right\rceil \right)$ possible merges. Hence $\mathcal{O}\left(n|\Sigma|B\right)$ amortized time.

\vspace{\baselineskip}

\noindent Analogue to~\citep{palmer_pac_paper} and~\citep{balle_2012} we can estimate the transition probabilities and arrive at the following theorem:

\begin{customtheorem}
    Given a target PDFA $A$ and given our streaming strategy along with our merge heuristic. Given the required input parameters $L', n', \mu', \Sigma, \epsilon, \delta$, then for each $L'\geq L$, $\mu' \leq \mu$, and $n' \geq n$ our framework returns a hypothesis automaton $H$ such that with probability at least $1-\delta$ the error is bounded by $L_1(H, A)\leq \epsilon$.  
\end{customtheorem}

\noindent The proof is analogous to~\citep{balle_2012}, ~\citep{palmer_pac_paper}, and \citep{balle_thesis}. Note that in this framework the parameters of the sketches $\alpha, \beta, \gamma$ can be inferred from $\delta$ and $\epsilon$. The framework above this has been shown to be a PAC learner, however it is not efficient due to the following constraint: The run-time of the similarity check is exponential in size of the alphabet $|\Sigma|$. We proposed two solutions to overcome this weakness, i.e. modeling unconditional probabilities and the MinHash-approach. With these however we cannot bound the error arbitrarily small anymore, unless we allow for hashing larger than the original size of the trace, which makes them inefficient again. We also practically deal with this weakness by splitting up the sketches into multiple sketches per state. If the similarity test fails for the 1-gram sketch, we can omit subsequent tests and so on. The theoretical upper worst case run-time still exists however.

\subsection{Making the PAC-learner efficient}

In order to be able to give an efficient PAC-learner we propose a slight alternative. In this version, instead of retrieving the distributions before performing similarity tests, we do test cell by cell. Consider two CMS as two matrixes $A, B \in \mathbb{R}^{m \times n}$. Rather than retrieving the counts for the symbols we perform the statistical tests on entry pairs $(A[1, 1], B[1, 1])$, $(A[1, 2], B[1, 2])$, $(A[2, 1], B[2, 1])$, and so on. In practice we empirically found through simple simulations that this approach does work worse on smaller sample sizes, however it seems to approach the same performance on larger sample sizes, given that the dimensions of the sketches are sufficiently large. This follows intuition given the nature of the Hoeffding-test, in which with large sample sizes the threshold shrinks, and hence it becomes more strict while having more evidence. In this case, all the collisions in the sketches can do is obstruct differences in distributions, which will be more likely to laid bare with increasing dimensions of the sketches, both in depth and width. Additionally, states that behave similar are expected to have similar CMS in all of their entries. The advantage of this approach to compare two states is that now we can compare in time $\mathcal{O}(w\dot d)$, regardless of $F_s$.

\vspace{\baselineskip}

\noindent An example: Say we have two baskets of three possible items. The probabilities of drawing a specific item in basked 1 are $p_1(item 1)=0.5$, $p_1(item 2)=0.2$, and $p_1(item 3)=0.3$, and those of basket 2 are $p_2(item 1)=0.5$, $p_2(item 2)=0.3$, and $p_3(item 3)=0.2$. The can see that two baskets differ in the distributions of item 2 and item 3. We want to test whether the two basket are identical and use CMS with one row for that. Now item 2 and item 3 get mapped into the same cell, hence both distributions under collision become $p(cell 1) = 0.5$ and $p(cell 2) = 0.5$. Hence the real distributions are different, but the sketches hide the differences by mapping them together. In the following we want to bound this type of error. At first we need a model for the expected collisions per row:

\begin{customcorollary}
    Given a sufficiently large alphabet $|\Sigma| >> w$ we can give a bound on the number of collisions $n'$ via $P(n' \leq t) = \sum_{x=0}^t \mathcal{N}(\mu_X, \sigma_X)$, where $\mathcal{N}(\mu_X, \sigma_X)$ is the normal distribution with mean $\mu_X$ and standard deviation $\sigma_X$.
\end{customcorollary}

\noindent This corollary is a direct consequence of the central-limit theorem. $\mu_X$ and $\sigma_X$ can be obtained from the assumption that the hash-functions of the sketches follow a uniform distribution with $p=\frac{1}{w}$, and the collisions can then be modeled via a Binomial distribution. We call the number of collisions corresponding with the desired bound on the probability $n_c$. 

\vspace{\baselineskip}

\noindent Having obtained a bound on the number of collisions, we next need a bound on the error per cell. We can obtain this via the Hoeffding-bound, i.e.

\begin{equation}
    P\left( \left| \frac{D_q(s)}{|q|} - \frac{D(s)}{|q|}\right| \geq \frac{\mu}{4n_c} \right) \leq e^{-2m_0\left( \frac{\mu}{4n_c} \right)^2}.
\end{equation}

\noindent The total bounded error per cell is thus $n_c e^{-2m_0\left( \frac{\mu}{4n_c} \right)^2}$, which we want to satisfy 

\begin{equation}
     n_c e^{-2m_0\left( \frac{\mu}{4n_c} \right)^2} \leq \frac{\delta'}{2|\Sigma|^2 n^2 m_0},
\end{equation}

\noindent assuming the number of futures is $F_s=1$. Following the approach in~\citep{palmer_pac_paper} we do get 

\begin{equation}
    m_0 \geq \frac{n_c n^2|\Sigma|^2}{\left( \frac{\mu}{4n_c} \right)^2\delta'}.
\end{equation}

\noindent With this we are able to give guarantees on the error bound per state, and the run-time of the Hoeffding-test is bounded by $\mathcal{O}(wd)$. The rest of the proof follows analogue to above.

\section{Discussion and conclusion}\label{section:discussion}

We propose a new strategy for streaming state machine learning using Count-Min-Sketches and the ability to undo previous mistakes. We demonstrate the efficacy of our approach on the PAutomaC dataset and show how our algorithm can fall under formal PAC bounds.
%At first we looked at the heuristic individually. 
We found that on this dataset the lookahead improvements start saturating at about $2$ to $3$ steps ahead into the future of each state respectively. We then showed that CSS and the variant CSS-MinHash perform well. We compared with the SpaceSave heuristic and point its subpar performance to the consistency heuristic, as discussed in Section \ref{sec:space_save_discussion}, and we point out the limitation that it does not support the undo-operation. Lastly we compared CSS with Alergia and showed similar performance when we enhance Alergia with the k-tails algorithm. An advantage of CSS here is that Alergia does require the prefix tree to be complete. After that we showed the effectiveness of our new streaming strategy in comparison with the traditional approach. It is clear that the new streaming strategy greatly improved results. It also showed the effectiveness of our CSS heuristic, as can especially be seen in the comparison of Alergia and CSS in the old streaming scenario. This discrepancy does however become smaller in the new streaming strategy. We expect our heuristic to make better choices mainly on the fringes of an incomplete prefix tree. Both streaming strategies greatly improved run-time and memory footprint compared with the batched version. On top of that our second streaming strategy in conjunction with our heuristic has been shown to satisfy formal PAC-bounds. Limitations of our work are the limited size of the dataset, and limitations to our hashing using MinHash, which mitigates but does not yet solve the potential run-time bottleneck completely.

\acks{This work is supported by NWO TTW VIDI project 17541 - Learning state machines from infrequent software traces (LIMIT).}

\bibliography{pmlr-sample}

\end{document}